\documentclass[sigconf]{acmart}

\usepackage{booktabs} 

\usepackage{xspace}
\usepackage{algorithm}
\usepackage{algpseudocode}

\graphicspath{{.}{images/}}

\newtheorem{theorem}{Theorem}
\newtheorem{lemma}{Lemma}
\newtheorem{problem}{Problem}
\newtheorem{definition}{Definition}
\newtheorem{corollary}{Corollary}

\DeclareMathOperator*{\argmin}{\arg\!\min}
\DeclareMathOperator*{\argmax}{\arg\!\max}

\newcommand{\mpara}[1]{\medskip\noindent{\bf{#1}}}
\newcommand{\spara}[1]{\smallskip\noindent{\bf{#1}}}

\newcommand{\reals}{\ensuremath{\mathbb{R}}\xspace}
\newcommand{\integers}{\ensuremath{\mathbb{N}}\xspace}
\newcommand{\bigO}{\ensuremath{\mathcal{O}}\xspace}
\newcommand{\np}{\ensuremath{\mathbf{NP}}\xspace}
\newcommand{\poly}{\ensuremath{\mathbf{P}}\xspace}

\newcommand{\X}{\ensuremath{X}\xspace}
\newcommand{\distance}{\ensuremath{d}\xspace}
\newcommand{\facilities}{\ensuremath{F}\xspace}
\newcommand{\nofacilities}{\ensuremath{m}\xspace}
\newcommand{\clients}{\ensuremath{C}\xspace}
\newcommand{\noclients}{\ensuremath{n}\xspace}
\newcommand{\sol}{\ensuremath{S}\xspace}
\newcommand{\optimal}{\ensuremath{O}\xspace}
\newcommand{\opt}{\ensuremath{o}\xspace}
\newcommand{\cost}{\ensuremath{\mathit{cost}}\xspace}
\newcommand{\afacility}{\ensuremath{s}\xspace}
\newcommand{\bfacility}{\ensuremath{t}\xspace}
\newcommand{\facilityi}{\ensuremath{s_i}\xspace}
\newcommand{\facilityj}{\ensuremath{s_j}\xspace}
\newcommand{\aclient}{\ensuremath{c}\xspace}

\newcommand{\x}{\ensuremath{x}\xspace}

\newcommand{\matA}{\ensuremath{\mathbf{A}}\xspace}
\newcommand{\matB}{\ensuremath{\mathbf{B}}\xspace}

\newcommand{\matD}{\ensuremath{\mathbf{D}}\xspace}
\newcommand{\vecx}{\ensuremath{\mathbf{x}}\xspace}

\newcommand{\vecones}{\ensuremath{\mathbf{1}}\xspace}

\newcommand{\servedset}[1]{\ensuremath{N_{#1}}\xspace}

\newcommand{\kmedian}{\ensuremath{k}-median\xspace}
\newcommand{\newkmedian}{{\sc recon}-\ensuremath{k}-{\sc median}\xspace}

\newcommand{\dks}{{\sc d$k$s}\xspace}
\newcommand{\mpd}{{\sc mpd}\xspace}
\newcommand{\lsalgo}{{\tt LocalSearch}\xspace}

\newcommand{\spectralcomps}{{\ensuremath\gamma}\xspace}

\newcommand{\twitterdata}{{\tt Twitter}\xspace}
\newcommand{\congressdata}{{\tt Congress}\xspace}
\newcommand{\domainsdata}{{\tt Domains}\xspace}

\copyrightyear{2019}
\acmYear{2019} 
\setcopyright{iw3c2w3}
\acmConference[WWW '19]{Proceedings of the 2019 World Wide Web Conference}{May 13--17, 2019}{San Francisco, CA, USA}
\acmBooktitle{Proceedings of the 2019 World Wide Web Conference (WWW'19), May 13--17, 2019, San Francisco, CA, USA}
\acmPrice{}
\acmDOI{10.1145/3308558.3313475}
\acmISBN{978-1-4503-6674-8/19/05}

\usepackage{balance}

\begin{document}

\title[Reconciliation $k$-median]{Reconciliation $k$-median: Clustering with Non-polarized Representatives}

\author{Bruno Ordozgoiti}
\authornote{The work described in this paper was done while this author was a visitor at Aalto University.}
\affiliation{%
  \institution{Department of Computer Systems \\
    Universidad Polit\'ecnica de Madrid}
}
\email{bruno.ordozgoiti@upm.es}

\author{Aristides Gionis}
\affiliation{%
  \institution{Department of Computer Science \\
   Aalto University}
}
\email{aristides.gionis@aalto.fi}

\renewcommand{\shortauthors}{B.\ Ordozgoiti \& A.\ Gionis}

\begin{abstract}
  
We propose a new variant of the \kmedian problem, 
where the objective function models not only the cost of assigning 
data points to cluster representatives, 
but also a penalty term for disagreement among the representatives. 
We motivate this novel problem by applications where we are interested in clustering 
data while avoiding selecting representatives that are too far from each other. 
For example, we may want to summarize a set of news sources, 
but avoid selecting ideologically-extreme articles in order to reduce polarization.


To solve the proposed \kmedian formulation we adopt the local-search algorithm of Arya et al.~\cite{arya2004local}, 
We show that the algorithm provides a provable approximation guarantee, 
which becomes constant under an assumption on the number of points for each cluster.
We experimentally evaluate our problem formulation and proposed algorithm
on datasets inspired by the motivating applications. 
In particular, we experiment with data extracted from Twitter, the US Congress voting records, and popular news sources. 
The results show that our objective can lead to choosing less polarized groups
of representatives 
without significant loss in representation~fidelity.

\end{abstract}

%
%



\maketitle

\section{Introduction}
\label{sec:intro}

Consider the problem of summarizing a set of news articles on a given topic. 
A standard approach to this problem is clustering: 
design a distance function that captures similarity between the news articles
and apply a clustering algorithm on the resulting metric space.
Common clustering formulations, such as \kmedian or $k$-means, 
can be used~\cite{jain1988algorithms}.
The original set of input news articles can then be summarized
by the (small) set of cluster representatives. 
In some cases, however, we may be interested
in selecting cluster representatives that are not too far from each other.
For example, we may want to find 
a set of representative news articles
that are not too extreme
so that they can provide a basis for constructive deliberation.
This motivation is similar to recent proposals in the literature
that aim to reduce polarization in social media~\cite{garimella2017reducing}
and balance the users' content consumption~\cite{garimella2017balancing}.
In this work we are interested in developing computational methods for 
clustering data
in a way that the disagreement of the cluster representatives is minimized.

Another motivating example appears in the context of electing a $k$-\emph{member committee} 
to represent a set of individuals, 
such as the employees of an organization or the members of a political party. 
Assuming that all individuals have public opinions on a set of issues, 
clustering these individuals on an opinion-representation space
will give a committee that faithfully represents 
the set of individuals with respect to the issues under consideration. 
Despite providing a good representation, however, 
a committee elected with a standard clustering approach 
may fail to reach consensus due to potential heterogeneity within the committee members.
Heterogeneity within elected members of an assembly is a widely acknowledged problem in politics ---
for instance, division of representatives often results in paralysis in various left-wing political formations.\footnote{https://www.theguardian.com/commentisfree/2019/feb/19/podemos-spanish-politics}
As in the previous example, 
we are interested in electing a committee
in a way that the disagreement of the elected members is minimized
while ensuring a faithful representation of the constituents.

Motivated by the previous examples we introduce a new
formulation of the \kmedian problem, 
where in addition to the \kmedian objective we also 
seek to minimize disagreement between the cluster representatives. 
As it is customary, we consider a metric space $(\X,\distance)$, 
where $\distance$ is a distance function for objects in $\X$.
We distinguish two subsets of $\X$, 
the set of facilities $\facilities$ and the set of clients $\clients$.
The goal is to select a set of $k$ facilities $\sol \subseteq \facilities$ 
--- the cluster representatives ---
so as to minimize the overall cost
\begin{equation}
\cost(\sol) =  
\sum_{\aclient\in\clients}\min_{\afacility\in\sol} \{\distance(\aclient,\afacility)\} + 
\frac{\lambda}{2} 
\sum_{\facilityi\in \sol}\sum_{\facilityj\in \sol} \distance(\facilityi,\facilityj).
\label{equation:objective}
\end{equation}
The first term is the same as in the standard \kmedian, 
and expresses the cost of serving each client by its closest selected facility. 
The second term is the one introduced in this paper
and expresses disagreement between cluster representatives. 
The parameter $\lambda$ determines the relative importance of the two terms.
Despite clustering being one of the most well-studied problems in statistical data analysis, 
and the numerous formulations and problem variants that have been studied in the literature, 
to our knowledge, the problem defined above has not been considered before.

As expected, the problem defined by optimizing Equation~(\ref{equation:objective}) is \np-hard;
in fact optimizing each of the two terms separately is an \np-hard problem.
%
Given the hardness of the problem, it is compelling to consider algorithms with provable 
approximation guarantees. 
For the \kmedian algorithm several approximation algorithms exist~\cite{charikar1999constant,li2016approximating}.
A local-search algorithm, which is simple to implement and scalable to large datasets, 
has been proposed by Arya et al.~\cite{arya2004local}.
The algorithm starts with an arbitrary solution and considers a swap of $p$ selected facilities
with $p$ non-selected facilities; the swap is materialized if the cost improves, 
and the process continues until the cost cannot be improved. 
Arya et al.\ show that this 
algorithm achieves an approximation guarantee equal to $3+2/p$, 
and the running time is $\bigO(n^p)$.
In particular, for $p=1$, the algorithm gives an approximation guarantee equal to $5$, 
while the running time is linear.

In this paper we show how to adapt the local-search algorithm of Arya et al.~\cite{arya2004local}
for the problem we consider and obtain an approximation guarantee $\bigO(k)$ 
in the case $F=C$.
The proposed algorithm considers 1-facility swaps, i.e., $p=1$.
Furthermore, when the clusters of the obtained solution have equal size of
$\Omega(\lambda k)$, the approximation factor becomes $11$, i.e., a constant. 
We complete the analysis of the proposed problem by deriving bounds on
the objective function.

Our contributions in this paper are summarized as follows.

\begin{itemize}
\item
We introduce the {\em reconciliation} $k$-{\em median} problem, 
a novel clustering problem formulation
where we aim to optimize the data representation cost
plus a term for agreement between the cluster representatives. 
\item 
We adapt the local-search algorithm of Arya et al.~\cite{arya2004local}
and obtain provable approximation guarantees for the proposed clustering problem.
\item 
We run experiments on datasets extracted from the Twitter social
network, US Congress voting records, and popular news sources. The results show that the proposed
objective can lead to the choice of less polarized groups of
representatives, as measured by a well-known method for ideology
estimation~\cite{barbera2014birds} and an objective estimate of the political leaning of news sources.
\end{itemize}

\spara{Erratum}: A previous version of this paper incorrectly omitted some of the necessary assumptions. In particular, the statement of Theorem \ref{theorem_k} overlooked the requirement that the client and the facility sets be the same. On the other hand, an argument in the proof of Theorem \ref{the:approx} required the cardinality of the clusters to be similar (and them being the same is sufficient). These details have been corrected in the present version. The proof of Theorem \ref{theorem_k}, which was omitted in the conference publication due to space constraints, is included here as well.

The rest of the paper is structured as follows.
In Section~\ref{section:related} we present a brief overview to 
the literature that is most related to our work.
In Section~\ref{section:problem} we formally define the 
{reconciliation} $k$-{median} problem. 
In Section~\ref{sec:algorithms} we present 
the local-search algorithm and state its approximability properties.
In Section~\ref{sec:experiments} we present our experimental evaluation, 
while Section~\ref{section:conclusions} is a short conclusion.
To improve readability, the hardness proof of the reconciliation \kmedian problem
and the proof of the approximation guarantee of the local-search algorithm 
are presented in the Appendix. 

\section{Related work}
\label{section:related}

Data clustering is one of the most well-studied problems in data analysis,
with applications in a wide range of areas~\cite{jain1988algorithms}.
Among the numerous formulations that have been proposed, 
in this paper we focus on the \kmedian problem setting, 
which has been studied extensively in the theoretical computer-science literature. 
Charikar et al.~\cite{charikar1999constant} gave the first constant-factor
approximation algorithm for the \kmedian problem, 
followed by improvements that relied on both
combinatorial and LP-based algorithms~\cite{charikar1999improved,jain2002new,jain2001approximation}. 
In this paper we build upon the local-search algorithm of Arya et al.~\cite{arya2004local}.
This is a simple-to-implement and scalable algorithm
that had been offering the best performance guarantee for over a decade.
The current best approximation guarantee is $2.67+\epsilon$, 
provided by the algorithm of Byrka et al.~\cite{byrka2014improved},  
which optimizes a part of the algorithm of Li and Svensson~\cite{li2016approximating}. 
However, the algorithm is not practical.

Variants of the $k$-median problem have also been considered, 
including the Euclidean $k$-median~\cite{ahmadian2017better},
capacitated $k$-median~\cite{byrka2014bi}, 
ordered $k$-median~\cite{byrka2018constant-factor}, 
and more.
To our knowledge, however, this is the first work to study the problem of \kmedian clustering
with a penalty on the disagreement of the cluster representatives.
Instead, researchers have studied the problem of 
selecting $k$ points to \emph{maximize} the sum of pairwise distances, 
i.e., the \emph{dispersion} of the selected point set. 
Several constant-factor approximation algorithms have been proposed
for the maximum-dispersion problem~\cite{fekete2004maximum,hassin1997approximation}. 
However, the maximization makes the problem different 
and it is not clear how to adapt those algorithms in our setting.

One of our motivating applications 
is summarization of social-media content with the aim of reducing polarization 
and balancing the information content delivered to users. 
This is a relatively new research area that is receiving increasing interest~\cite{barbera2014birds,garimella2017reducing,garimella2017balancing,munson2013encouraging,musco2017minimizing}.
However, to the best of our knowledge, none of the proposed approaches uses a clustering formulation.

The second motivating application is election of committees and representatives. 
In some cases, election questions can also be formulated as voting problems.
Voting in general has been studied extensively in social sciences and economics literature. 
From the algorithmic perspective, researchers have studied questions about 
voting in social networks and concepts such as liquid and viscous democracy~\cite{boldi2009voting,yamakawa2007toward}.
In addition to being not directly related to our paper, this line of work 
does not directly model agreement between elected representatives.

\section{Problem formulation}
\label{section:problem}

\subsection{Preliminaries} 

We formulate our problem in the general setting of \textit{metric facility location}~\cite{jain2001approximation}. 
We consider a metric space $(\X,\distance)$ and two subsets $\facilities,\clients \subseteq \X$, 
not necessarily disjoint. 
The set \facilities represents \emph{facilities}, 
and the set \clients represents \emph{clients}. 
A special case of interest is when clients and facilities
are defined over the same set, i.e., $\facilities=\clients$.
In our discussion we consider the more general case that the sets of clients and facilities are disjoint.
The function $\distance : \X\times\X\rightarrow\reals$ is a distance measure between pairs of points in \X.
When $\aclient\in\clients$ and $\afacility\in\facilities$, 
the distance $\distance(\aclient,\afacility)$ represents the cost
of serving client~\aclient with facility \afacility. 
The number of facilities is denoted by $\nofacilities=|\facilities|$ and 
the number of clients by $\noclients=|\clients|$.

The goal is to open $k$ facilities --- i.e., choose $k$ points in $\facilities$ --- 
such that the cost of serving each client in $\clients$ with the nearest \emph{selected} facility is minimized. 
Given a set of facilities $\sol\subseteq \facilities$, with $|\sol|=k$, 
and $\afacility\in \sol$, we use $\servedset{\sol}(\afacility)$ 
to denote the set of clients served by facility $\afacility$ in the solution \sol, that is, 
$\servedset{\sol}(\afacility)=\{\aclient\in \clients \mid \afacility=\argmin_{\x\in\sol}\distance(\aclient,\x)\}$. 
In the \emph{facility location} formulation 
each facility has an associated cost, which is incurred if the facility is opened (selected).
The objective is to minimize the total cost of serving all clients plus the cost of 
opened facilities, while there is no restriction on the number of opened facilities.
When the cost of opening each facility is zero and it is required to open at most $k$ facilities, 
the problem is known as \emph{\kmedian}.

\subsection{Reconciling cluster representatives} 

The problems described above, facility location and \kmedian, 
are commonly used to find cluster representatives without regard to the relative position
of the representatives themselves. 
As discussed in the introduction, our goal is to modify the problem definition 
so as to find solutions in which the cluster representatives are close to each other. 
To achieve our goal we propose the following clustering variant, 
which we name reconciliation \kmedian.

\begin{problem} [\newkmedian]
\label{def:prob_clust}
Given a metric space $(\X,\distance)$, two sets $\facilities,\clients\subseteq \X$, 
$k \in \integers$, and a real number $\lambda>0$, find 
a set $\sol\subseteq\facilities$ with $|\sol|=k$, so as to 
\emph{minimize} the cost function
\begin{equation}
\label{equation:objective-2}
\cost(\sol) =  
\sum_{\afacility\in\sol}\sum_{\aclient\in\servedset{\sol}(\afacility)} \distance(\aclient,\afacility) + 
\frac{\lambda}{2}
\sum_{\facilityi\in \sol}\sum_{\facilityj\in \sol} \distance(\facilityi,\facilityj).
\end{equation}
\end{problem}


\smallskip
In order to characterize the hardness of this problem, 
we analyze the two terms of the objective in isolation. 
The first term, which results from setting $\lambda=0$, 
is equivalent to the classical metric \kmedian problem, 
shown to be \np-hard by Papadimitriou~\cite{papadimitriou1981worst}. 
To analyze the second term, we define the following equivalent problem, 
which asks to find a subset of $k$ points that minimize the sum of pairwise distances
in a metric space, i.e., minimum pairwise distances (\mpd).
\begin{problem}[\mpd]
\label{def:prob_mpd}
Given a metric space $(\X,\distance)$, a subset of objects in the metric space 
$\facilities=\{\x_1, \dots, \x_n\} \subseteq \X$,
and a number $k\in \integers$, with $k < n$, 
define a matrix $\matA \in \reals^{n\times n}$ as 
$\matA_{ij}=\distance(\x_i,\x_j)$ for all $\x_i,\x_j \in \facilities$. 
The goal is to find a binary vector \vecx of dimension $n$ that has exactly $k$ coordinates equal to 1
and minimizes the form $\vecx^T\matA\,\vecx$. 
In other words, we want to find
\begin{align*}
\min             \quad & \vecx^T\matA\,\vecx, \\
\mbox{subject to}\quad & \vecx\in \{0,1\}^n \mbox{ and }\, \vecx^T\vecones=k.
\end{align*}
\end{problem}
The following lemma establishes that 
there exists no polynomial-time algorithm to find the exact solution
to the \mpd problem, unless $\poly=\np$.
The proof is given in the Appendix.
\begin{lemma}
\label{lemma:mdp-np-hardness}
Problem \mpd is \np-hard.
\end{lemma}
%
Lemma \ref{lemma:mdp-np-hardness} establishes that optimizing separately the second term of 
the objective function~(\ref{equation:objective-2}) is an \np-hard problem.
Note, however, that the hardness of the two terms in Problem \ref{def:prob_clust} 
does not immediately imply the hardness of the overall problem. 
Consider, for instance, an objective of the form $\min_\vecx \left\{ f(\vecx) + (c-f(\vecx)) \right\}$. 
Even though optimizing $f$ can be an arbitrary \np-hard problem, the overall problem has a constant value,
and thus, it is trivial to optimize --- there is nothing to be done. 
We now show that Problem~\ref{def:prob_clust} is indeed \np-hard.
\begin{lemma}
\label{lemma:newkmedian-np-hardness}
Problem \newkmedian is \np-hard.
\end{lemma}
\begin{proof}
Consider an instance of the \mpd problem, for a given set \facilities and a number $k$, 
and form an instance $(\facilities,\clients,k,\lambda)$ of the \newkmedian problem,  
where \clients is any arbitrary set of clients with $\facilities\cap\clients=\emptyset$.  
Set 
$\distance(\afacility,\aclient)=\frac{1}{2}\max_{\facilityi,\facilityj \in \facilities}\{\distance(\facilityi,\facilityj)\}$, 
for all $\aclient\in \clients$ and $\afacility \in \facilities$. 
Note that the distance function \distance is still a valid metric.
We have that $\sum_{\aclient\in \clients} \min_{\afacility\in\sol} \distance(\aclient,\afacility)$ is constant
for any potential solution set $\sol\subseteq\facilities$,
which implies that optimizing \newkmedian is 
equivalent to optimizing \mpd.
Thus, from Lemma \ref{lemma:mdp-np-hardness} we obtain that \newkmedian is \np-hard.
\end{proof}

\subsection{Bounds on the objective} 

To complete the analysis of the \newkmedian problem we offer bounds on the objective,
which can be used to evaluate the quality of the solution obtained by any algorithm for the problem
at a given instance.
For this to be useful, the bounds need to be non-trivial and as close as possible to the optimal solution. 
We will first show how to obtain a lower bound on the second term of the objective, 
that is, $\sum_{\facilityi,\facilityj\in \sol} \distance(\facilityi,\facilityj)$. 
We first introduce the following definition. 
\begin{definition} 
Given two sequences of real numbers $\lambda_1 \geq \lambda_2 \geq
\dots \geq \lambda_n$ and $\mu_1 \geq \mu_2 \geq \dots \geq \mu_m$,
with $m<n$, we say that the second sequence interlaces the first if  
\begin{equation*}
\lambda_i \geq \mu_i \geq \lambda_{n-m+i} ~~~ \mbox{for $i=1,\dots,m$}
\end{equation*}
\end{definition}
We will employ the following result from Haemers~\cite{haemers1995interlacing}.
\begin{lemma}
\label{the:interlacing}
Let \matA be a symmetric $n\times n$ matrix, and let \matB be a principal submatrix of \matA.
Then the eigenvalues of matrix \matB interlace those of matrix \matA.
\end{lemma}
We now state the following result.
\begin{theorem}
Let \matD be the pairwise distance matrix of facilities of an instance of \newkmedian problem.
Define matrix $\widetilde\matD$ so that $\widetilde\matD_{ij}=\sqrt{\matD_{ij}}$, that is, 
a matrix whose entries are the square roots of the entries of \matD. Let $\lambda_i(\widetilde\matD)$ denote the $i$-th absolutely largest
eigenvalue of $\widetilde\matD$.
Then
\begin{equation*}
k\lambda_1(\matD) \geq \sum_{\facilityi\in \sol}\sum_{\facilityj\in \sol} \distance(\facilityi,\facilityj) 
\geq \sum_{i=n-k+1}^n \lambda_i^2(\widetilde\matD).
\end{equation*}
\end{theorem}
\begin{proof}
Recall that for any real symmetric matrix \matD, 
$\|\matD\|_F^2=\sum_{i}\sigma_i^2(\matD)=\sum_{i}\lambda_i^2(\matD)$. 
It is easily seen that if $\vecx$ is a binary vector with exactly $k$ entries equal to 1, 
$\vecx^T\matD\,\vecx$ is equal to the sum of the entries of a $k\times k$ principal submatrix of \matD, 
which is in turn equal to the squared Frobenius norm of the corresponding submatrix of~$\widetilde\matD$. 
Combined with Lemma~\ref{the:interlacing}, this proves the lower bound.

The upper bound follows immediately from the variational characterization of the eigenvalues 
\citep[chap. 9.2]{kreyszig1978introductory}.
\end{proof}

A lower and upper bound on the first term of the objective can simply be given by 
$\sum_{\aclient\in\clients}\min_{\afacility\in\facilities}\distance(\aclient,\afacility)$
and
$\sum_{\aclient\in\clients}\max_{\afacility\in\facilities}\distance(\aclient,\afacility)$,
respectively.
This is useful only for problem instances where the number of facilities is relatively small
compared to the number of clients.

\section{The local-search algorithm}
\label{sec:algorithms}

In this section we present the proposed algorithm for the \newkmedian problem
and state its properties. 
The algorithm uses the local-search strategy,  
proposed by Arya et al.~\cite{arya2004local}, 
but adapted for the objective function of \newkmedian. 
The algorithm starts with an arbitrary solution consisting of $k$ selected facilities.
It then proceeds in iterations. 
In each iteration it considers whether it is possible to swap a selected facility
with a non-selected facility and obtain an improvement in the objective score. 
If such an improvement is possible, the corresponding swap is performed. 
The algorithm terminates when no such swap is possible.
At each point during its execution, 
the algorithm maintains a set of $k$ clusters over the set of clients, and a selected facility for each cluster, 
defined by assigning each client to its closest selected facility. 
Pseudocode of this local-search procedure is given in Algorithm~\ref{algorithm:local-search}.

\begin{algorithm}[t]
\caption{Local search\label{algorithm:local-search}}
\begin{algorithmic}[1]
\Procedure{\lsalgo}{\facilities, \clients, $k$, $\lambda$}
\label{alg:local_search}
\State $\sol \gets $ random subset of \facilities of cardinality $k$
\State converged $\gets$ {\tt false}
\While {not converged}
\If {there exist $\afacility\in\sol$ and $\bfacility\in\facilities\setminus\sol$ such that 
$\cost(\sol\setminus\{\afacility\}\cup \{\bfacility\})<\cost(\sol)$}
\State $\sol \gets \sol\setminus\{\afacility\}\cup \{\bfacility\}$
\Else 
\ 
converged $\gets$ {\tt true}
\EndIf
\EndWhile
\State return $\sol$
\EndProcedure
\end{algorithmic}
\end{algorithm}

For the analysis 
we denote by $\sol=\{\afacility_1,\ldots,\afacility_k\}\subseteq\facilities$ 
the solution returned by \lsalgo, 
and $\optimal=\{\opt_1,\ldots,\opt_k\}\subseteq\facilities$ 
an optimal solution. 
As mentioned before
we use the notation 
$\servedset{\sol}(\afacility)$ to denote the set of clients that are assigned to facility \afacility in the solution \sol, 
and 
$\servedset{\optimal}(\opt)$ to denote the set of clients that are assigned to facility \opt
in the optimal solution \optimal. 

To analyze the performance of \lsalgo we follow the ideas of Arya et al.~\cite{arya2004local}.
The proofs are included in the Appendix.

As a result, in the case $F=C$, the \lsalgo algorithm yields a $\bigO(\lambda k)$-factor
approximation guarantee on the quality of the solution achieved. 
%
%
\begin{theorem}
  \label{theorem_k}
Let $(\clients,\clients,k,\lambda)$ be an instance of the \newkmedian problem. 
Let \sol be a solution returned by the \lsalgo algorithm, 
and let \optimal be an optimal solution. Then
\begin{equation}
\cost(\sol) \leq 2(\lambda k+5) \cost(\optimal).
\end{equation}
\end{theorem}

Furthermore, we are able to improve the analysis 
and obtain an approximation guarantee that does not depend 
on the number of facilities $k$. 
For the improved result we need to make the mild assumption that
the number of clients in any cluster of the optimal solution and the solution
returned by the algorithm is $\Omega(\lambda k)$.
In particular, we have.
\begin{theorem}
\label{the:approx}
Let $(\clients,\clients,k,\lambda)$ be an instance of the \newkmedian problem. 
Let \sol be a solution returned by the \lsalgo algorithm, 
and let \optimal be an optimal solution.
Assume that 
$|\servedset{\sol}(\afacility)| \geq \lceil 2\lambda \rceil k$,
$|\servedset{\optimal}(\opt)| \geq \lceil 2\lambda \rceil k$
and
$|\servedset{\sol}(\afacility)| = |\servedset{\optimal}(\opt)|$
  for all $\afacility\in\sol$, $\opt\in\optimal$. 
Then
\begin{equation*}
\cost(\sol) \leq \max\{11, 4\lambda\}\, \cost(\optimal).
\end{equation*}
\end{theorem}

The running time of the \lsalgo algorithm is $\bigO(\noclients\nofacilities k)$ per iteration. 
For most applications $k$ is considered to be a constant. 
When the number of facilities \nofacilities is of the same order of magnitude with the number of clients \noclients, 
e.g., in the important special case $\facilities=\clients$, 
the algorithm has quadratic complexity per iteration.
However, in many applications the number of facilities is significantly smaller than
the number of clients. Thus, we expect that the algorithm is very efficient in practice for those cases.

We also note that Arya et al.~\cite{arya2004local} show how to perform swaps of $p$ facilities simultaneously
and obtain an improved performance guarantee at the expense of increased running time. 
In our case, the penalty term in our objective, 
which captures the disagreement among the cluster representatives, 
makes the analysis significantly more complex
and it is not clear how to use simultaneous swaps in order to 
achieve a similar quality-performance trade-off.
Thus, this direction is left for future work.

\section{Experimental evaluation}
\label{sec:experiments}

We perform experiments to assess the proposed concept of clustering with reconciliation of representatives, 
as well as the performance of the proposed \lsalgo algorithm.\footnote{Our implementation of the algorithm is available at \url{https://github.com/brunez/recon-kmedian-ls}}
Our objective is to evaluate by some objective measure 
whether the proposed problem formulation, as well as natural variations, 
can lead to a choice of representatives or sources that are more moderate, less polarizing and more likely 
to reach consensus.

To enrich the experimental setting and produce a more interesting set of empirical observations, we relax some of the requisites of our theoretical results. Namely, observe that in order to prove the approximation guarantees of the \lsalgo algorithm, we require that the distance function satisfy the properties of a metric, and that it be the same for facilities and clients. However, we believe that in practical scenarios, one might benefit from considering a wider set of options, especially if we consider the exploratory nature of clustering algorithms. Therefore, we measure dissimilarity between objects using functions that do not necessarily qualify as metrics, and we consider different ones for facilities and clients.

At a high level, our experimental methodology is as follows: 
We start with a dataset for which clients and facilities model a 
natural clustering problem and for which a distance function \distance is available. 
For the facilities of the dataset we seek to obtain a polarity score $\pi$, 
which is independent of the distance function \distance:
facilities with similar polarity scores $\pi$ are more likely to agree. In addition, facilities with scores closer to the middle of the spectrum are less likely to disseminate extremist ideologies. 
We then apply our clustering algorithm with varying values of the parameters $k$ and $\lambda$.
We are interested in answering the following questions:

\begin{itemize}
\item[{\bf Q1.}]
How does the agreement between selected representatives or the polarization of information sources
(measured by the independent polarity score $\pi$) 
change as a function of $\lambda$? 
In other words, can we get more reconciled representatives or less extreme sources by increasing
the weight of the disagreement penalty term
(second term of the objective function~(\ref{equation:objective-2}))?

\item[{\bf Q2.}]
How does the \kmedian score change as a function of $\lambda$? 
In other words, can we find solutions with more reconciled representatives
but without significant loss in representation fidelity
(i.e., first term of the objective function~(\ref{equation:objective-2}))?

\item[{\bf Q3.}]
What is the impact of the parameter $k$ on both polarity score and \kmedian score?
\end{itemize}

\begin{table*}[t]
\begin{center}
  \caption{Summary of the datasets.}
  \vspace{-3mm}
  \label{tab:datasets}
  \centering
  \begin{tabular}{lrrc}
    \toprule
    \multicolumn{1}{c}{Name}  & \multicolumn{1}{c}{Number of}   & \multicolumn{1}{c}{Number of} & \multicolumn{1}{c}{Distance functions} \\
                              & \multicolumn{1}{c}{clients}     & \multicolumn{1}{c}{facilities} & \\
    \midrule
    \twitterdata &  $3\,302\,362$ & $500$ & Shortest path, Spectral embedding $+$ Euclidean    \\
    \congressdata &   $420$  & $420$ & Euclidean   \\
    \domainsdata &   $6\,104$  & $469$ & Weighted Jaccard $+$ Mentions, Latent space $+$ Euclidean   \\
    \bottomrule
  \end{tabular}
\end{center}
\end{table*}

For our experimental evaluation we use the following datasets.\footnote{%
The datasets are available at https://doi.org/10.5281/zenodo.2573954}

\smallskip
\noindent
{\twitterdata}: The dataset, obtained by Lahoti et al.~\cite{lahoti2018joint},
consists of a set of politically active Twitter accounts. 
We remove stubs --- i.e., accounts that follow only one account and have no followers --- 
resulting in $3\,302\,362$ accounts. 
Out of those we consider 500 popular ones --- with at least $50\,000$ followers --- 
as candidate facilities, that is, representatives. 
As remarked in the beginning of section \ref{sec:experiments}, we can extend the proposed framework by considering different metrics for the two terms of the objective function. 
This corresponds to a practical setting where the agreeability of the selected representatives 
is measured differently than their affinity to their respective consituents. 
Specifically, for this dataset we consider the following distance functions.

\begin{enumerate}
\item \textbf{Facility-Client:}
We compute distances between facilities and clients
as the length of the shortest path between two Twitter 
accounts in the undirected follower Twitter graph.\footnote{The follower graph corresponds to a snapshot taken in July 2017.}

\item \textbf{Facility-Facility:}
To compute distances between facilities we use shortest-path distances, as before. 
We also use Euclidean distances on the spectral embedding with $\spectralcomps$ components, 
as described by Belkin et al.~\cite{belkin2002laplacian}. 
We scale the resulting distance matrix so that the average of all entries is equal 
to that of the shortest-path distance matrix. 
This way we ensure that the magnitude of $\lambda$ has an equivalent effect using the different metrics.
\end{enumerate}

\smallskip
\noindent
{\congressdata}: We collect roll call voting records from the present US Congress 
using the public domain \texttt{congress} API.\footnote{https://github.com/unitedstates/congress/wiki} 
We build a dataset where each row corresponds to a Congress representative 
and each column is a binary variable representing the issue being voted. 
Missing values are imputed using class-conditional means, 
where the classes we consider are the two parties: democrats and republicans. 
``{\em Present}'' and ``{\em Not voting}'' votes are considered to be missing. 
We omit votes where all representatives are missing. 
We also omit representatives for whom we could not obtain an ideology estimate using the approach described below, 
or who missed too many votes. 
For this dataset, we use the Euclidean distance between the vectors
corresponding to each of the representatives. To make this experiment
closer to a plausible practical scenario, we restrict half of the
facilities to be democrats and the other half to be republicans. In
addition, clients are served by the closest facility of the same party.

\smallskip
\noindent
{\domainsdata}:  We combined the domain-related data described in the work of Bakshy et al. \cite{bakshy2015exposure} with the \twitterdata dataset. The set of facilities consisted of 469 domains hosting the news sources most often shared on the Facebook social network. The client set is comprised of 6\,104 of the most politically active Twitter users. We consider two alternatives for computing the distances.

\begin{enumerate}
\item \textit{Mentions}: Given a facility $f$ and a client $c$, let $n_{cf}$ be the number of times a tweet by $c$ contains a link to $f$. Then $d(f,c)=(n_{cf} + 1)^{-1}$. To compute the pairwise distances between facilities we do the following. Consider two facilities, $f$ and $g$. Let $S_f$ (respectively $S_g$) be the set of clients that have tweeted a link to $f$ (respectively $g$) at least once. 
We define $W = \sum_{c \in S_f\cup S_g}\left ( \log n_{cf}\mathbb I \{c\in S_f\} + \log n_{cg} \mathbb I \{c\notin S_f\}\right )$, where $\mathbb I$ is the indicator function. 
Then
\begin{equation}
    \label{eq:jacc_dist}
    d(f,g) = 1-\frac{\sum_{c \in S_f\cap S_g}\log n_{cf}}{W}.
\end{equation}
We define $\log 0=0$. Since the objective function of \newkmedian (Equation (\ref{def:prob_clust})) sums over all ordered pairs of facilities, this distance function is in effect symmetric when applied to our problem.
Note that this is akin to the Jaccard index for set similarity, but each element is weighted with a measure of its relevance.

\item 
\textit{Latent}: We construct a matrix $A$ such that $A_{ij}$ is the number of times a tweet by user $j$ contains a link to domain $i$. We compute the singular value decomposition $A=U\Sigma V^T$ and extract the latent representation for both domains and users in the first 9 components (which account for 50\% of the total Frobenius norm of $A$). If $k=9$, domain $i$ is represented as $U_{i,:k}$ and user $j$ as $V_{i,:k}$. To compute both facility-facility and 
facility-client distances we take the Euclidean distances between the corresponding latent representations.
\end{enumerate}

The characteristics of the datasets are summarized in Table~\ref{tab:datasets}.

For all datasets, in order to compute the objective of \newkmedian we take averages instead of the sums of distances. Note that this amounts to scaling both sums, so it is equivalent to setting $\lambda$ to a particular value. The advantage of taking averages is that $\lambda$ has an impact at small values, i.e., at ``small'' factors of 1.

\spara{Ground truth polarity scores ($\pi$).} 
To measure the polarity scores of the facilities we employ different methods depending on the dataset.

In the case of \twitterdata and \congressdata we use the approach described by Barber\'a~\cite{barbera2014birds}, 
which estimates the ideological leaning of a Twitter account as a real value. 
For \twitterdata, we use polarity scores collected at the same time as the follower graph (July 2017). 
For \congressdata, we collected the ideological estimates in May 2018. 
Using this method and the proposed datasets, all the elicited polarity scores are between -3 and 3. 
We measure the polarity of the chosen representatives as follows. 
Given a solution $\sol=\{\afacility_1, \dots, \afacility_k\}$, 
let $\pi(\afacility_i)$ denote the estimated polarity of facility $\afacility_i$. 
We define the polarity of solution $\sol$ as the sample standard deviation of the set 
$\{\pi(\afacility_1), \dots, \pi(\afacility_k)\}$.

For \domainsdata, we use the ideological leaning score associated to each domain as described in the work of Lahoti et al.~\cite{lahoti2018joint}. These scores were computed roughly as the fraction of interactions (visits or shares) by conservative users, out of total interactions. We translated the scores so that they fall between -0.5 (left) and 0.5 (right). In this scenario, we are interested in choosing less polarized sources. We therefore measure the polarity of the chosen set as the $\ell$-2 norm of the vector $(\pi(\afacility_1), \dots, \pi(\afacility_k))$.

\begin{figure*}[t]
  \centering
  \includegraphics[width=\textwidth]{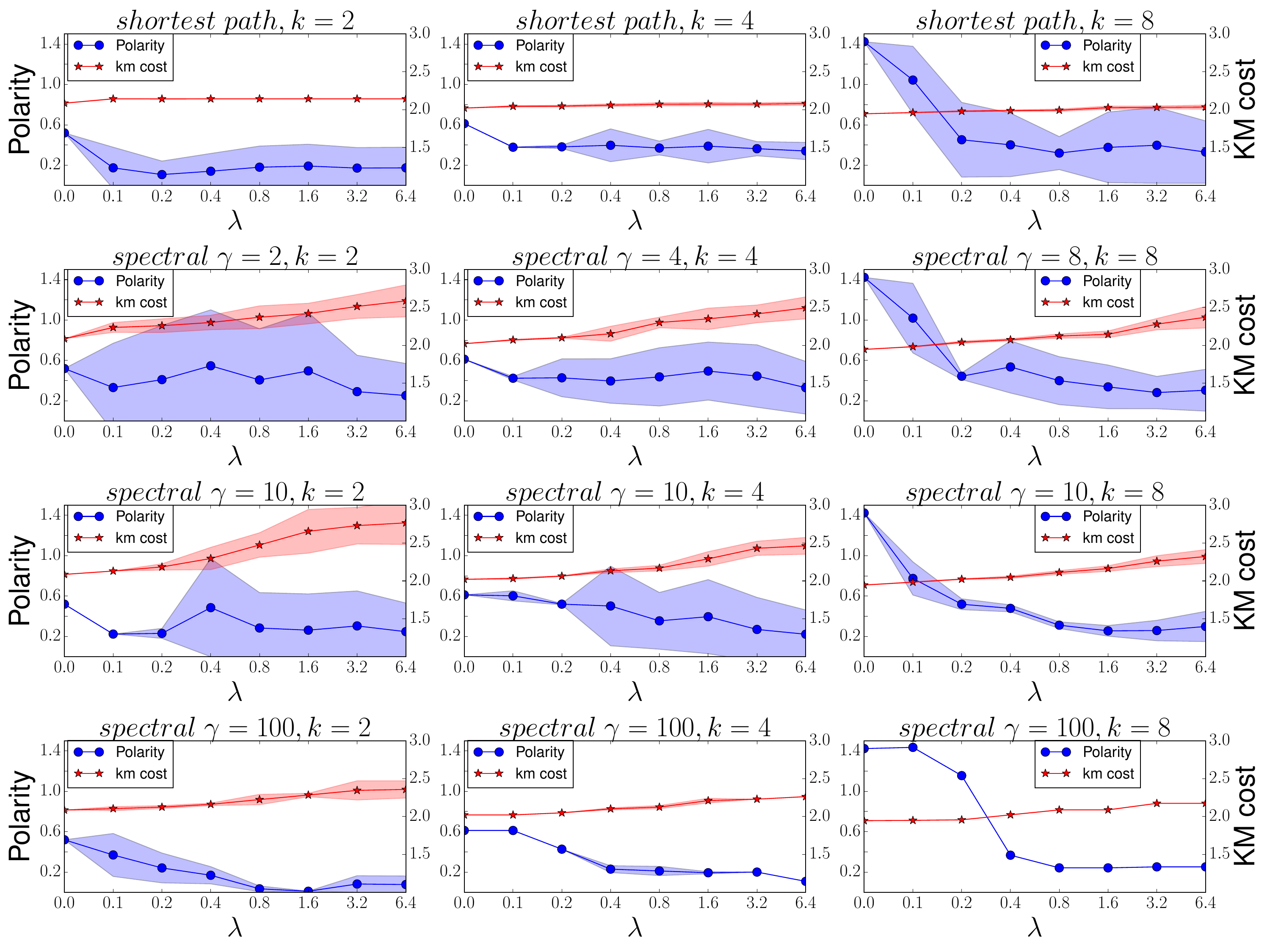}
  \caption{\label{fig:twitter}Results on the \twitterdata dataset with different metrics and values of $k$.}
\end{figure*}

\begin{figure*}[t]
  \centering
  \includegraphics[width=\textwidth]{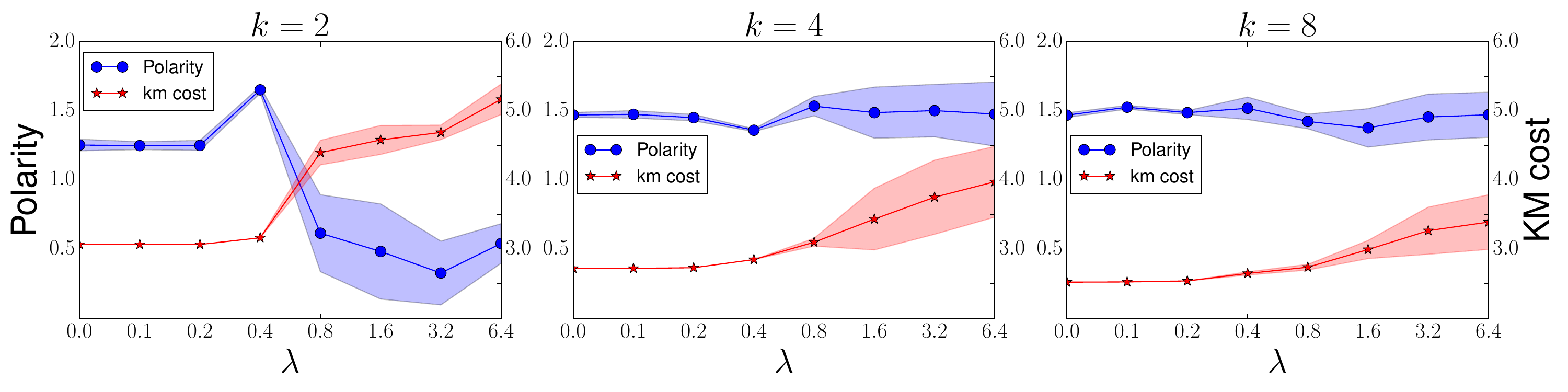} 
  \caption{\label{fig:congress}Results on the \congressdata dataset for different values of $k$.} 
\end{figure*}

\begin{figure*}[t]
  \centering
  \includegraphics[width=\textwidth]{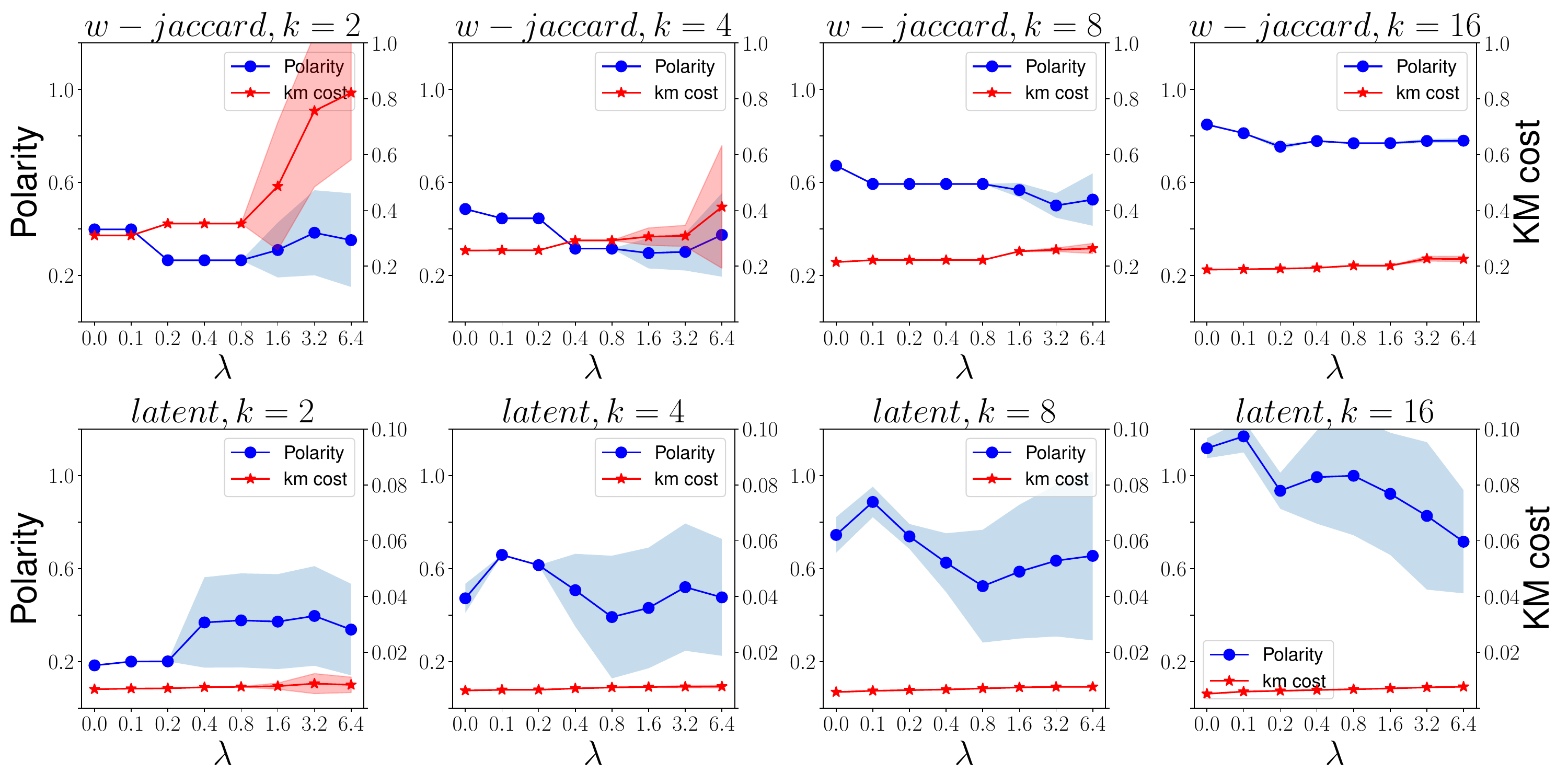}
  \caption{\label{fig:domains}Results on the \domainsdata dataset with different metrics and values of $k$.}
\end{figure*}

\spara{Results on \twitterdata dataset.}
We run the algorithm on the \twitterdata dataset setting the number of selected facilities to $k=2,4,8$ and 
$\lambda=0,2^{i}/10$ for $i=1,\dots,6$. 
For computing the pairwise distances between facilities we use either the shortest path metric 
or the spectral embedding with $\gamma=k,10,100$. 
Figure~\ref{fig:twitter} illustrates the results. 
We depict, as a function of $\lambda$, the polarity of the chosen representatives, 
measured as described above, along with the $k$-median cost of the solution --- 
i.e., the first term of the objective function~(\ref{equation:objective-2}).
We run the algorithm 40 times for each setting and report the average and standard deviation bands. 
We can see that increasing the value of $\lambda$ leads to significantly less polarized representatives in various cases. 
The effect is most noticeable for larger values of $k$, and 
particularly consistent using the spectral embedding with 100 components. 
An interesting result on this dataset is that we can achieve significant decreases in polarity 
without incurring much additional $k$-median cost. That is, 
it is possible to elect a much more agreeable committee --- with respect to the chosen polarity measure --- 
without notable loss in representation fidelity.

\spara{Results on \congressdata dataset.}
We run the algorithm on the \congressdata dataset, 
using the same configurations as for \twitterdata. 
Figure~\ref{fig:congress} illustrates the results. 
Here, the decrease in polarity is only clear in the case $k=2$, for values of $\lambda$ at least 0.8. 
It should be noted, however, that the voting data and the polarity scores come from completely different sources. 
It would therefore be interesting to carry out further experiments with these data.

\spara{Results on \domainsdata dataset.}
We run the algorithm on the \domainsdata dataset, using the same configurations as for \twitterdata but considering $k=16$ as well, as in the case of news sources it is practical to consider larger sets. Figure~\ref{fig:domains} shows the results, using the distance function defined in Equation (\ref{eq:jacc_dist}) (\textit{w-jaccard}) and the latent representation (\textit{latent}). We run each configuration 80 times and report average results and standard deviation bands. The reduction in polarity is noticeable, in particular using the latent representation with larger values of $k$. For very small sets of news sources (e.g., $k=2$) the method does not exhibit a reduction in polarity. In order to gain further insight on the impact of the penalty term, we report an example of the news sources that appear more frequently as $\lambda$ increases. Specifically, we take a case where decrease in polarity is noticeable (\textit{latent}, $k=8$, $\lambda=0.8$). We then collect the 16 sources that appear the most, and do the same for $\lambda=0$. The results are shown in Table~\ref{tab:topdomains}. For each domain, we report the frequency (i.e., the fraction of times it was part of the solution out of the 80 runs), the ideological score and the number of times it was mentioned in the collected tweets.

\begin{table*}[t]
  \begin{center}
      \caption{Top-16 news sources for different values of $\lambda$, using latent representations.}
      \vspace{-3mm}
      \label{tab:topdomains}
      \begin{tabular}{lrrr}
        \toprule
        \multicolumn{4}{c}{$\lambda = 0$}   \\ \midrule 
        Domain & Frequency & Ideology & Mentions  \\ \midrule
        nydailynews.com       & 1.000 & -0.114 & 10\,191  \\
        politico.com          & 1.000 & -0.073 & 36\,184  \\
        slate.com             & 1.000 & -0.341 & 14\,364  \\
        cbsnews.com           & 0.687 & -0.057 &  4\,394  \\
        buzzfeed.com          & 0.687 & -0.262 & 11\,683  \\
        twitchy.com           & 0.537 &  0.469 & 13\,192  \\
        westernjournalism.com & 0.462 &  0.450 &  3\,562  \\
        9news.com             & 0.462 & -0.016 &     349  \\
        politifact.com        & 0.350 & -0.240 &  3\,097  \\
        cbsloc.al             & 0.337 & -0.081 &  2\,526  \\
        christianpost.com     & 0.337 &  0.337 &     383  \\
        theatlantic.com       & 0.312 & -0.176 &  6\,883  \\
        newrepublic.com       & 0.312 & -0.335 &  1\,626  \\
        lifenews.com          & 0.200 &  0.483 &  3\,657  \\
        6abc.com              & 0.200 & -0.252 &     819  \\
        usatoday.com          & 0.112 & -0.064 & 18\,513 \\
        \bottomrule
      \end{tabular}
      \hspace{2mm}
      \begin{tabular}{lrrr}
        \toprule
        \multicolumn{4}{c}{$\lambda=0.8$} \\ \midrule   
        Domain & Frequency & Ideology & Mentions \\ \midrule
        chicagotribune.com & 0.962 & -0.082 & 1\,531 \\
        chron.com          & 0.475 &  0.170 &    431 \\
        abc13.com          & 0.325 &  0.005 &    255 \\
        9news.com          & 0.250 & -0.016 &    349 \\
        detroitnews.com    & 0.225 &  0.090 &    535 \\
        azc.cc             & 0.225 & -0.028 &    744 \\
        nbcwashington.com  & 0.225 & -0.214 &    485 \\
        csmonitor.com      & 0.225 & -0.030 &    382 \\
        wjla.com           & 0.225 & -0.160 &    374 \\
        msn.com            & 0.200 & -0.031 &    615 \\
        kgw.com            & 0.187 & -0.118 &    107 \\
        christianpost.com  & 0.175 &  0.336 &    383 \\
        abc7chicago.com    & 0.175 & -0.251 &    328 \\
        inquisitr.com      & 0.175 &  0.049 & 2\,150 \\
        stripes.com        & 0.175 &  0.182 &    555 \\
        wsbtv.com          & 0.137 & -0.043 &    167 \\
        \bottomrule
      \end{tabular}
  \end{center}
\end{table*}

\spara{Number of iterations.}
Even though the time complexity of the local-search algorithm per iteration 
is not too high, a legitimate concern to be raised is the possibility that
it might require a large number of iterations to converge. 
Our observations, however, suggest that in practice a small number of
iterations --- where by iteration we understand the inspection of all
candidate changes --- are necessary. Table \ref{tab:iterations} shows
the average and maximum number of iterations for the \twitterdata dataset. 

\begin{table}[t]
\begin{center}
  \caption{Number of iterations on \twitterdata dataset~(avg/max).}
  \vspace{-3mm}
  \label{tab:iterations}
  \centering
  \begin{tabular}{lrrc}
    \toprule
    \cmidrule(r){1-2}
    \multicolumn{1}{c}{Metric}  & \multicolumn{1}{c}{$k=2$}   & \multicolumn{1}{c}{$k=4$} & \multicolumn{1}{c}{$k=8$} \\                             
    \midrule
    Shortest path        &   $2.5 / 4$    &  $2.8/5$ &  $3.09/7$   \\
    Spectral, $\gamma=2$ &   $2.18/3$  & -  &  -  \\
    Spectral, $\gamma=4$ &   -  & $2.73/5$  &  -  \\
    Spectral, $\gamma=8$ &   -  & -  &  $3.13/6$  \\
    Spectral, $\gamma=10$ &   $2.32/4$  & $2.8/6$  &  $3.32/7$  \\
    Spectral, $\gamma=100$ &   $2.45/4$  & $2.75/5$  &  $2.58/4$  \\
    \bottomrule
  \end{tabular}
\end{center}
\end{table}

\section{Conclusions}
\label{section:conclusions}

We have considered the problem clustering data
so as to optimize the total representation cost
plus an additive term to penalize disagreement among the chosen representatives.
The proposed problem, which we name reconciliation \kmedian, 
has applications in summarizing data with non-polarized representatives, 
as well as in electing $k$-member committees that are more likely to reach consensus. 
We have shown the proposed problem to be \np-hard and derived bounds on the objective.
Inspired by the literature on related problems, we have analyzed a
local-search algorithm in this context and derived approximation guarantees, 
of factor $\bigO(\lambda k)$ in the general setting, and constant under mild assumptions. 
Through experiments on real data coming from a social network and voting records, we have shown
empirically how the proposed formulation can lead to the choice of less polarized groups of
representatives, as measured by a well-known method for ideology estimation, as well as less ideologically-extreme sets of news sources. 
This work opens various enticing directions for future inquiry. 
First it would be interesting to determine whether the
approximation guarantees can be improved, as well as to attempt to
find tight examples to know the possible extent of said improvement.
Second, it would be interesting to perform further experiments on
similar and other datasets. It is particularly compelling to
improve our understanding of how different metrics can interact with known
methods for estimating polarization.

\mpara{Acknowledgements.} This work was supported by
three Academy of Finland projects  (286211, 313927, and 317085), 
and the EC H2020 RIA project ``SoBigData'' (654024).

\section*{Appendix}
\subsection{Hardness results}
Here we prove Lemma~\ref{lemma:mdp-np-hardness}, 
which is a key ingredient of the proof of hardness for \newkmedian. 
Before we proceed, we provide a definition of Densest $k$-subgraph, 
a well-known \np-hard optimization problem which we employ in our reduction.
\begin{problem}(densest $k$-subgraph --- \dks)
 \label{def:prob_dks}
Given a simple graph $G=(V,E)$ with $|V|=n$ and adjacency matrix $A$, 
and a number $k\in \mathbb N$ with $k < n$, find 
\begin{align*}
\min_\vecx             \quad & \vecx^T\matA\,\vecx, \\
\mbox{subject to}\quad & \vecx\in \{0,1\}^n \mbox{ and }\, \vecx^T\vecones=k.
\end{align*}
\end{problem}

We can now prove the aforementioned lemma.
\begin{proof}[Proof of lemma \ref{lemma:mdp-np-hardness}]

We proceed by reduction from Densest $k$-subgraph (\dks).

Consider an instance of \dks, $G=(V,E)$, $|V|=n$ with adjacency matrix $A$. We define a matrix $\tilde A$ as follows:
\[
\tilde A_{ij} =  
   \begin{cases} 
      A_{ij} & (i,j) \in E \\
      0 & i=j \\
      2 & \mbox{otherwise}.
   \end{cases}
\]
Notice that this matrix is symmetric, the diagonal (and nothing else) is zero, and since $\min_{i\neq j}\tilde A_{ij}=1$ and $\max_{i\neq j}\tilde A_{ij}=2$, for all $i$, $j$, and $\ell$
\[
\tilde A_{ij}\leq \tilde A_{i,\ell}+\tilde A_{\ell,j}.
\]

We want to show that if
\[
\vecx=\argmin_{\substack{\vecx \in \{0,1\}^n\\\vecx^T \vecones=k}}\vecx^T\tilde A\, \vecx, 
\]
then
\[
\vecx=\argmax_{\substack{\vecx \in \{0,1\}^n\\\vecx^T \vecones=k}}\vecx^TA\, \vecx.
\]

We can write 
\begin{align*}
\vecx^T \tilde{A}\, \vecx=\sum_i\sum_j\tilde A_{ij}\mathbb I\{x_i=x_j=1\} = M+2N,
\end{align*}
 where we have defined
\begin{align*}
M=\left|\,\left\{(i,j) \mid x_i=x_j=1 \right\} \cap E\,\right|, & \mbox{ and} 
\\ N=\left|\,\left\{(i,j) \mid x_i=x_j=1 \right\} \cap \bar E\,\right|. & 
\end{align*}
That is, $M$ is the number of pairs in $\vecx$ with a corresponding edge in $G$, and $N$ is the number of pairs in $\vecx$ {\it without} a corresponding edge in $G$. 
Similarly, 
\begin{align*}
\vecx^T A\, \vecx=\sum_i\sum_j A_{ij}\mathbb I\{x_i=x_j=1\} = M.
\end{align*}


Furthermore, note that $M+N=2{k \choose 2}=k^2-k$.

It follows that for all vectors $\vecx\in\{0,1\}^n$, with $\vecx^T\vecones=k$, it is
\begin{align*}
\vecx^T \tilde{A}\,\vecx = -\vecx^T A\,\vecx +2(k^2-k).
\end{align*}

Since $2(k^2-k)$ is a constant, it follows that a vector $\vecx\in\{0,1\}^n$, with  $\vecx^T\vecones=k$,
minimizes $\vecx^T \tilde{A}\,\vecx$ (the \mpd objective)
\emph{if and only if} it maximizes   $\vecx^T {A}\,\vecx$ (the \dks objective).

\end{proof}

\subsection{Approximation guarantees}
Given a set of clients $C$, a set of facilities $F$, and any subset of~facilities $S\subseteq F$, 
we employ the following notation:

\begin{itemize}
\item[--] $f(S)=\sum_{x\in C}d(x,s(x))$, where $s(x)$ is the facility in $S$ assigned to client $x\in C$;

\item[--]  $g(S)=\frac{1}{2}\sum_{x\in S}\sum_{y\in S}d(x,y)$;

\item[--]  $\servedset{\sol}(s)$ is the set of clients served by facility $s$ in $S$.
\end{itemize}

We rely on the following crucial facts. First, a result from the work of Arya et al.~\citep{arya2004local}.
\begin{lemma}[\citep{arya2004local}]
  \label{lem:pairs}
Let $O$ be an optimal set of $k$ facilities and  $S$ an arbitrary set of $k$ facilities. There exists a set of $k$ pairs 
$\Sigma = \{(o_i,s_j)\mid i,j \in [k]\} \subseteq O\times S$  satisfying the following properties:
\begin{enumerate}
\item Every $o\in O$ is considered in exactly one pair.
\item Every $s \in S$ is considered in at most two pairs.
\item We can choose the set of pairs $\Sigma$ such that the following inequality holds:
\begin{equation}
\label{eq:arya_approx}
5f(O) - f(S) \geq \sum_{(o_i,s_j)\in\Sigma}\left ( f(S-s_j+o_i) - f(S) \right ).
\end{equation}
\end{enumerate}
\end{lemma}

Second, by definition of local optimality, 
for a locally optimal solution $\sol=\{s_1, \dots, s_k\}$ and an optimal solution $\optimal=\{o_1, \dots, o_k\}$, we have
\begin{equation}
\label{eq:local_opt}
  f(S-s_j+o_i) + \lambda g(S-s_j+o_i) \geq  f(S) + \lambda g(S),
\end{equation}
for any $i,j \in [k]$. 

We now have the ingredients for the proof of Theorem \ref{theorem_k}, from Section~\ref{sec:algorithms}.

\begin{proof}[Proof of theorem \ref{theorem_k}]
Since $\sol$ and $\sol-s_j+o_i$ differ in one facility only,
we have
\begin{align*}
g(\sol-s_j+o_i)- & g(\sol) \\
 & = \frac{1}{2}\sum_{x \in \sol-s_j+o_i}\sum_{y \in \sol-s_j+o_i}d(x,y) - \frac{1}{2}\sum_{x \in \sol}\sum_{y \in \sol}d(x,y) \\ 
 & = \sum_{\substack{s\in \sol\\ s\neq s_j}}d(o_i,s) - \sum_{s \in \sol}d(s,s_j).
\end{align*}

We consider a set of pairs $\Sigma$ satisfying the properties of lemma \ref{lem:pairs}. Summing the above difference over all $k$ pairs we get
\begin{align*}
\sum_{(o_i, s_j)\in \Sigma} (g(\sol & -s_j+o_i)- g(\sol)) \\
 & = \sum_{(o_i, s_j)\in \Sigma}\left (\sum_{\substack{s\in \sol\\ s\neq s_j}}d(o_i,s) - \sum_{s \in \sol}d(s,s_j) \right ).
\end{align*}

Therefore, summing inequality \eqref{eq:local_opt} over all pairs in $\Sigma$ we get
\begin{align*}
&\sum_{(o_i,s_j)\in\Sigma}\left ( f(\sol-s_j+o_i) - f(\sol) \right )
\\  &+\lambda\sum_{(o_i, s_j)\in \Sigma}\sum_{\substack{s\in \sol\\ s\neq s_j}}d(o_i,s) - \lambda\sum_{(o_i, s_j)\in \Sigma}\sum_{s \in \sol}d(s_j,s)\geq 0,
\end{align*}

and using inequality \eqref{eq:arya_approx} and rearranging we obtain
\begin{align}
  \label{eq:central_ineq}
  5f(\optimal) + \lambda\sum_{(o_i, s_j)\in \Sigma}\sum_{\substack{s\in \sol\\ s\neq s_j}}d(o_i,s) \geq f(\sol) + \lambda\sum_{(o_i, s_j)\in \Sigma}\sum_{s \in \sol}d(s_j,s).
\end{align}

We consider this inequality, and modify it so that it becomes dependent only on factors of $f(\optimal)$, $g(\optimal)$, $f(\sol)$, $g(\sol)$. To accomplish this, we will consider the set of pairs $\Sigma$.

First, let us examine the following quantity:
\[
\sum_{(o_i,s_j)\in \Sigma}\sum_{\substack{s\in \sol\\ s\neq s_j}}d(o_i,s).
\]

For each term $d(o_i, s), s\in \sol, (o_i,s_j)\in \Sigma$, we consider three cases:
\begin{itemize}
\item $s\in \optimal$. Then $d(o_i,s) \leq g(\optimal)$.
\item $s \notin \optimal$ and $s \in N_\optimal(o_i)$. Then $d(o_i,s) \leq f(\optimal)$.
\item $s \notin \optimal$ and $s \notin N_\optimal(o_i)$. Then $d(o_i, s) \leq d(o_i, o_j) + d(o_j, s)$, where $s\in N_\optimal(o_j)$, and thus $d(o_i,s) \leq f(\optimal) + g(\optimal)$.
\end{itemize}
Now, consider the sum \(\sum_{(o_i,s_j)\in \Sigma}\sum_{\substack{s\in \sol\\ s\neq s_j}}d(o_i,s)\)
and observe that
\begin{itemize}
\item for every $i$, $o_i$ is in $k-1$ terms of the sum;
\item for every $j$, $s_j$ is in at most $k$ terms of the sum.
\end{itemize}

Thus, 
\begin{align}
  \label{eq:ineq_kmedianterm}
  \sum_{(o_i,s_j)\in \Sigma}\sum_{\substack{s\in \sol\\ s\neq s_j}}d(o_i,s) \leq  k(g(\optimal) + f(\optimal)).
\end{align}

  We now examine the quantity 
\begin{equation}
\label{eq:sum_rhs_k}
\sum_{(o_i,s_j)\in \Sigma}\sum_{s \in \sol}d(s_j,s).
\end{equation}

In particular, we will show that
\begin{align}
  \label{eq:ineq_reconterm}
 2\sum_{(o_i, s_j)\in \Sigma}\sum_{s \in \sol}d(s_j,s) \geq g(\sol).
\end{align}

We consider the set $P\subseteq \sol$ of facilities in some swap, i.e.,
 $P=\{$%
$s\in \sol$ such that $(o,s)\in \Sigma$ for some $o\in \optimal$%
 $\}$. 
For any facility  $s\in P$, 
note that all the entries from $g(\sol)$ that involve terms $d(s,s')$ are present in
expression~(\ref{eq:sum_rhs_k}), and so we can trivially bound them.

Thus, we can assume that $k$ is even and $|P|=k/2$, that is, all
facilities in $\sol$ either are in two swaps from $\Sigma$, or in no
swaps at all.
We number the facilities in $P$ $1$ through $k/2$, and the rest
$k/2+1$ through $k$.

We bound each term of $g(\sol)$
corresponding to facilities $i,j$ not in $P$, $i<j$ as
\[
d(i,j) \leq d(i,j-i) + d(j-i,j).
\]
Note we can bound all such terms this way, with each term in the upper
bound appearing at most twice, and being part of the sum in expression~(\ref{eq:sum_rhs_k}). Thus, we get inequality~(\ref{eq:ineq_reconterm}).

Now, consider inequality \eqref{eq:central_ineq}, and apply inequalities \eqref{eq:ineq_kmedianterm} and \eqref{eq:ineq_reconterm}. We obtain
\begin{align}
  2\left(5f(\optimal) + \lambda k(g(\optimal) + f(\optimal))\right) \geq f(\sol) + \lambda g(\sol).
\end{align}

\end{proof}

\begin{proof}[Proof of theorem \ref{the:approx}]
We consider inequality \eqref{eq:central_ineq}, and modify it so that it becomes dependent only on factors of $f(O)$, $g(O)$, $f(S)$, $g(S)$. To accomplish this, we will consider the set of pairs $\Sigma$.

First, let us examine the following quantity:
\[
\sum_{(o_i,s_j)\in \Sigma}\sum_{\substack{s\in S\\ s\neq s_j}}d(o_i,s).
\]

By the triangle inequality, we have that the sum corresponding to each pair $(o_i,s_j)\in \Sigma$ is bounded as follows:
\begin{align}
\label{eq:tri_bound}
\sum_{\substack{s\in S\\ s\neq  s_j}}d(o_i,s) &\leq \sum_{\substack{s\in S\\ s\neq  s_j}}d(o_i,o_{h_s}) + d(o_{h_s},x) + d(x,s),
\end{align}
where $o_{h_s}$ is such that $\servedset{\optimal}(o_{h_s})\cap \servedset{\sol}(s)\neq \emptyset$ and $x \in \servedset{\optimal}(o_{h_s})\cap \servedset{\sol}(s)$. If $\servedset{\sol}(s)\subseteq \servedset{\optimal}(o_i)$, we consider  $d(o_i,s) \leq d(o_i,o_{h_s}) + d(o_i,x) + d(x,s)$, choosing $o_{h_s}$ as described below.


For notational convenience, we define 
\begin{align*}
\sigma(o_i) &= \sum_{\substack{s\in S\\ s\neq  s_j }}d(o_i,o_{h_s}) + d(o_{h_s},x) + d(x,s),
\end{align*}
for each $(o_i,s_j)\in \Sigma$ ---
note that each $o_i$ appears exactly
once in the pairs in the set $\Sigma$, so the pair $(o_i,s_j)$ is
uniquely determined by $o_i$.

We want to choose the entries of inequality (\ref{eq:tri_bound}) such that
\begin{itemize}
\item[$-$] every term of the form $d(o_i,o_{h_s})$ appears at most twice; 
\item[$-$] every term of the form $d(o_{h_s},x)$ appears once;  
\item[$-$] every term of the form $d(x,s)$ appears once. 
\end{itemize} 

To achieve this, we need a set of replacements such that for each $\sigma(o_i)$, the $k-1$ corresponding replacements contain $k-1$ distinct entries of the forms $d(o_i,o_{h_s})$, $d(o_{h_s}, x)$, $d(x, s)$. 

We define a function $\mu$ that maps each $x\in C$ to a pair $(o_i,s_j)\in \optimal\times\sol$ 
such that $x\in \servedset{\sol}(s_j)\cap \servedset{\optimal}(o_i)$.  
Since $|\servedset{\sol}(s_i)|\geq k$,
  $|\servedset{\optimal}(o_i)|\geq k$ and $|\servedset{\sol}(s)| =
|\servedset{\optimal}(o)|$ for all $s\in \sol, o \in \optimal$,  
we can choose a subset $\tilde C \subseteq C$ of $k^2$ points such that each $s \in
  S$, as well as each $o \in O$, appears in $k$ of the $k^2$ pairs
  associated to the chosen points.

We now argue that said choice is indeed possible.
Consider a bipartite multigraph, where the vertices of each partition
correspond to the facilities in $S$ and $O$, respectively. We add an
edge between $s$ and $o$ for each client served by the two facilities
in the corresponding solution.
That we can pick these $k^2$ clients satisfactorily follows from the
fact that every regular bipartite multigraph is
1-factorable \cite{akiyama2011factors}.

We can now conclude that 
\begin{itemize}
\item[$-$] every element of the form $d(o_i, o_j)$ can only appear either in $\sigma(o_i)$ or in $\sigma(o_j)$. Since for all $\sigma(o_i)$, the entry $d(o_i, o_j)$ is unique, each $d(o_i, o_j)$ appears at most twice;
\item[$-$] for any element of the form $d(o, x), o\in O$, $x$ is unique;
\item[$-$] for any element of the form $d(s, x), s\in S$, $x$ is unique.
\end{itemize}

We have established the following inequality
\begin{equation}
  \label{eq:ineq_g}
\sum_{(o_i, s_j)\in \Sigma}\sum_{\substack{s\in S\\ s\neq  s_j}}d(o_i,s) \leq 2g(O) + f(O) + f(S).
\end{equation}

Considering each cluster in both $S$ and $O$ to have at least $\lceil 2\lambda\rceil k$ points, we can extend the result to account for $\lambda$. To do this, we repeat the above choice of clients $\lceil 2\lambda \rceil$ times  and obtain
\begin{equation}
  \label{eq:ineq_glambda}
2\lambda \sum_{(o_i, s_j)\in \Sigma}\sum_{\substack{s\in S\\ s\neq  s_j}}d(o_i,s) \leq 4\lambda g(O) + f(O) + f(S).
\end{equation}
To see this, note that for every entry of the form $d(o, x)$ or $d(s, x)$, $x$ can still be chosen to be unique.

Now, consider inequality \eqref{eq:central_ineq} again. Multiplying it by $2$ we get
\begin{align*}
  10f(O) + 2\lambda\sum_{(o_i, s_j)\in \Sigma}\sum_{\substack{s\in S\\ s\neq s_j}}d(o_i,s) \geq 2f(S) + 2\lambda\sum_{(o_i, s_j)\in \Sigma}\sum_{s \in S}d(s_j,s).
\end{align*}
By inequalities \eqref{eq:ineq_glambda} and \eqref{eq:ineq_reconterm} we obtain the following:
\begin{align*}
  & 10f(O) + 4\lambda g(O) + f(O) + f(S) \geq 2f(S) + \lambda g(S)
  \\ & \Rightarrow  10f(O) + 4\lambda g(O) + f(O) \geq f(S) + \lambda g(S)
  \\ & \Rightarrow  11f(O) + 4\lambda g(O) \geq f(S) + \lambda g(S).
\end{align*}
Therefore, 
\begin{align*}
\max\{11,4\lambda\}cost(O) \geq  cost(S). \tag* \qedhere
\end{align*}

\end{proof}

\clearpage
\bibliographystyle{ACM-Reference-Format}
\bibliography{citations}
\balance 

\end{document}